\documentclass[a4paper,UKenglish,hyperref,cleveref, autoref, thm-restate]{lipics-v2021}
\title{Characterising Robust Instances of Ultimate Positivity for Linear Dynamical Systems}
\titlerunning{Robust Ultimate Positivity}

\author{Mihir Vahanwala}{Max Planck Institute for Software Systems, Saarland Informatics Campus, Germany}{mvahanwa@mpi-sws.org}{https://orcid.org/0009-0008-5709-899X}{}

\authorrunning{M.\ Vahanwala} 

\Copyright{Mihir Vahanwala} 



\nolinenumbers

\usepackage{amsfonts,stmaryrd,mathtools}
\usepackage{amsmath}
\usepackage{amssymb}
\usepackage{wrapfig}
\usepackage[ruled,vlined,linesnumbered,titlenotnumbered,noend]{algorithm2e}
\usepackage{xcolor}
\usepackage{booktabs}   
\usepackage{subcaption} 
\usepackage{pgf}
\usepackage{tikz}
\usepackage{listings}
\usepackage[colorinlistoftodos,prependcaption,textsize=tiny]{todonotes}
\usepackage{float}
\usepackage{tcolorbox}
\usepackage{bbm}
\usepackage{hyperref}
\usepackage{cleveref}
\usepackage{lineno}

\newcommand{\naturals}{\mathbb{N}}
\newcommand{\integers}{\mathbb{Z}}
\newcommand{\reals}{\mathbb{R}}

\usepackage{graphicx}

\begin{document}
\maketitle

\begin{abstract}
Linear Dynamical Systems, both discrete and continuous, are invaluable mathematical models in a plethora of applications such the verification of probabilistic systems, model checking, computational biology, cyber-physical systems, and economics. We consider discrete Linear Recurrence Sequences and continuous C-finite functions, i.e.\ solutions to homogeneous Linear Differential Equations. The Ultimate Positivity Problem gives the recurrence relation and the initialisation as input and asks whether there is a step $n_0$ (resp.\ a time $t_0$) such that the Linear Recurrence Sequence $u[n] \ge 0$ for $n > n_0$ (resp.\ solution to homogeneous linear differential equation $u(t) \ge 0$ for $t > t_0$). There are intrinsic number-theoretic challenges to surmount in order to decide these problems, which crucially arise in engineering and the practical sciences. In these settings, the difficult corner cases are seldom relevant: tolerance to the inherent imprecision is especially critical. We thus characterise \textit{robust} instances of the Ultimate Positivity Problem, i.e.\ inputs for which the decision is locally constant. We describe the sets of Robust YES and Robust NO instances using the First Order Theory of the Reals. We show, via the admission of quantifier elimination by the First Order Theory of the Reals, that these sets are semialgebraic. 
\end{abstract}

\section{Introduction}
Discrete Linear Dynamical Systems are essentially captured by Linear Recurrence Sequences (LRS). An LRS $(u[n])_{n\in \naturals}$ of order $k$ is a sequence satisfying a recurrence relation 
$$
u[n+k] + c_{k-1}u[n+k-1] + \dots + c_1u[n+1] + c_0u[n] = 0
$$
for all $n \in \naturals$. It is uniquely specified by $(c, v) \in \reals^{2k}$ where the coefficients $c = (c_0, c_1, \dots, c_{k-1})$ and the initial terms $v = (u[0], u[1], \dots, u[k-1])$. In the continuous setting, the dynamics are described in terms of derivatives. We use $f^{(k)}$ to denote the $k^{th}$ derivative of a $k$-times differentiable function $f$. We take $f^{(0)} = f$. Our fundamental model of the continuous setting is the C-finite function. A C-finite function $u: \reals_{\ge 0} \rightarrow \reals$ of order $k$ is a function satisfying a homogeneous linear differential equation
$$
u^{(k)}(t) + c_{k-1}u^{(k-1)}(t) + \dots + c_1u^{(1)}(t) + c_0u^{(0)}(t) = 0
$$
for all $t \in \reals_{\ge 0}$. It is uniquely specified by $(c, v) \in \reals^{2k}$ where the coefficients $c = (c_0, c_1, \dots, c_{k-1})$ and the initialisation $v= (u^{(0)}(0), u^{(1)}(0), \dots, u^{(k-1)}(0))$.

Given the ubiquity of LRS and C-finite functions in mathematics, science, engineering and economics, the decision problems associated with them are naturally important. These include the Positivity Problem (is the sequence or function always non-negative?), the Ultimate Positivity Problem (does there exist an iterate or point in time, beyond which the sequence or function is non-negative?), and the Skolem Problem (does the sequence or function have a zero in some interval?). Unfortunately, a general procedure to decide these problems seems to elude contemporary number theory. The authors of \cite{mignotte} and \cite{vereshchagin} independently proved the Skolem problem to be decidable for LRS of order four. Positivity and Ultimate Positivity are decidable for LRS of order five, but decidability at orders six and higher would entail significant number-theoretic breakthroughs \cite{joeljames3}. The spectral restriction to \textit{simple} LRS results in the decidability of Positivity at order nine \cite{ouaknine2014positivity}, and the full decidability of Ultimate Positivity \cite{ouaknine2014ultimate}. The state of the art on the continuous front is the decidability of Ultimate Positivity at order eight, and number-theoretic hardness at orders nine and higher \cite{continuous2023}.

\paragraph*{Motivation and Related Work}
Decidability is conventionally discussed in the context of precise rational or real algebraic input. A common theme is that the number-theoretic hardness rarely raises its head: hard instances constitute a set of zero Lebesgue measure. Practitioners need contend with the imprecision and uncertainty inherent in their settings, and would sensibly err on the side of caution in scenarios where making a mathematically sound decision is infeasible. These arguments have motivated the recent attention to the \textit{robust} variants of the problem: intuitively, they ask whether the decision remains the same despite small perturbations to the input.

The works \cite{originalstacs,originalarxiv,pseudo21} investigate discrete LRS. They keep the recurrence fixed, and ask whether there is an $\varepsilon$ such that the given initialisation satisfies the decision problem with tolerance $\varepsilon$: i.e.\ $\varepsilon$ perturbations to the initialisation (either once at the start, or at every step of the trajectory) must result in sequences that still satisfy the decision problem. This paradigm uses logic to navigate the spectral intricacy that it runs into.

On the other hand, \cite{N21} introduced an even more encompassing notion of robustness: it considers a model of computation that accepts arbitrary real numbers in both the recurrence and initialisation components of the input. This implicitly builds robustness into the setting: a terminating decision procedure cannot read the input to infinite precision. This means that decisions can only be made for instances where the answer is locally constant, and finite precision suffices to detect this: we call these robust instances. Recently, \cite{neumann2023} showed it can be efficiently detected when the input is a Robust YES or a Robust NO instance of Ultimate Positivity for C-finite functions. The detection procedure can be adapted to LRS as well. 

\paragraph*{Our contribution}
In this paper, we consider the problem of Ultimate Positivity in both the discrete and continuous settings. We concern ourselves with the theoretical characterisation of Robust YES and Robust NO instances. Formally:
\begin{definition}[Robust Ultimate Positivity]
An instance $(c, v) \in \reals^{2k}$ of Ultimate Positivity is robust if there exists an $\varepsilon$-neighbourhood of $(c, v)$ such that an arbitrary $(c', v')$ in the neighbourhood is a YES instance if and only if $(c, v)$ is.
\end{definition}

\begin{figure}[h]
\centering
\includegraphics[width=0.85\textwidth]{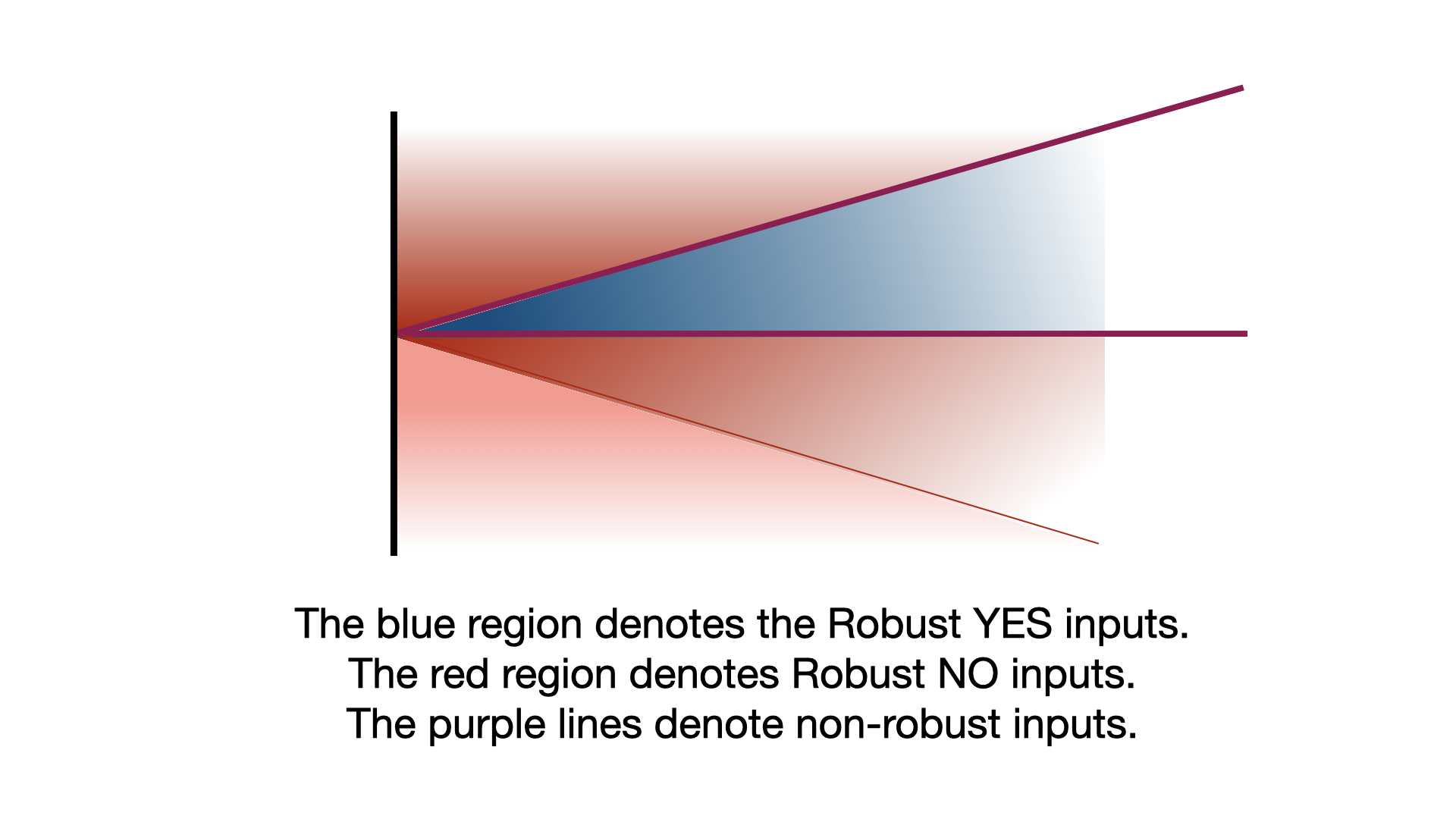}

\caption{Intuiting the underlying topology}
\label{fig:intuition}
\end{figure}

Figure \ref{fig:intuition} illustrates the underlying topology. In the measure theoretic sense, almost all inputs are Robust Instances: those in the relative interior of the sets of YES and NO instances. Our main technical result is that the sets of robust instances are semialgebraic.

\begin{theorem}[Main result]
For all $k \in \naturals$, the sets of Robust YES and Robust NO instances of the Ultimate Positivity Problem for both discrete and continuous Linear Dynamical Systems of order $k$ are semialgebraic.
\end{theorem}

We observe that by the very definition of robustness, these sets must be open and full-dimensional, and thus testing membership, in principle, only involves checking finitely many polynomial \textit{inequalities}. This, indeed, is the underlying reason that enables the maximal partial decision procedure for the Ultimate Positivity Problem for real input in \cite{neumann2023}: our contribution lies in the explicit intuition, discovery, and exposition of this reason.\footnote{We do acknowledge, however, that the techniques in \cite{neumann2023} are more optimised for implementation.}

\paragraph*{Mathematical tools: Semialgebraic sets and the first order theory of the reals}
We briefly initiate the unfamiliar reader to the important terminology. A set $S \subseteq \reals^d$ is semialgebraic if it is the finite union of sets defined by equalities and inequalities involving polynomials from $\integers[X_1, \dots, X_d]$. It can be expressed in disjunctive normal form as 
$$
\bigvee_i \bigwedge_j f_{i,j}(x_1, \dots, x_d) \sim_{i,j} 0
$$
where each $\sim_{i, j}$ is either $>$ or $=$.

Let $\sigma(x_1, \dots, x_d, y_1, \dots, y_m)$ be a Boolean combination of atomic predicates of the form $g(x_1, \dots, x_d, y_1, \dots, y_m) \sim 0$, where $\sim$ is either $>$ or $=$, and $g \in \integers[X_1, \dots, X_d, Y_1, \dots, Y_d]$. A formula of the First Order Theory of the Reals has the form
$$
Q_1 y_1 \dots Q_m y_m.~ \sigma(x_1, \dots, x_d, y_1, \dots, y_d)
$$
where each $Q_i$ is a quantifier, either $\exists$ or $\forall$. The variables $x_1, \dots, x_d$ are called free variables, while $y_1, \dots, y_m$, which occur under the scope of a quantifier, are called bound variables. The truth of a formula is evaluated under an \textit{assignment} to the free variables. Two formulae with the same set of free variables are equivalent if the same assignment to both formulae always yields the same truth value. Tarski \cite{tarski} showed that the First Order Theory of the Reals admits quantifier elimination; more precisely:
\begin{theorem}[Tarski]
\label{theorem:tarski}
For every formula $$\Psi(x_1, \dots, x_d) \equiv Q_1 y_1\dots Q_m y_m. ~ \sigma(x_1, \dots, x_d, y_1, \dots, y_m)$$ with free variables $x_1,\dots, x_d$, one can compute an equivalent quantifier-free formula $\Phi(x_1, \dots, x_d)$.
\end{theorem}
We note that the quantifier-free $\Phi$ can be brought in disjunctive normal form, and thus precisely fits our definition of the characterisation of a semialgebraic set. Thus, our proof strategy is to encode the sets of Robust YES and Robust NO instances using the First Order Theory of the Reals, and then leverage Theorem \ref{theorem:tarski} to prove that these sets are indeed semialgebraic.

\paragraph*{High level synopsis}
We now outline the implementation of our strategy. The exponential polynomial closed form is now a standard perspective to reason about LRS and C-finite functions. In \S \ref{section:exppoly}, we discuss how it is particularly indicative of the asymptotic behaviour of our systems. While the exponential polynomial itself is not continuous in the input, our crucial asymptotic properties indeed are. An especially elegant way to prove this is by applying control-theoretic generating functions introduced in \S \ref{section:genfunc} for a transformed perspective. We combine our observations to demonstrate the First Order definability of the sets of robust instances of Ultimate Positivity in \S\ref{section:leadterm}. Finally, in \S\ref{section:discussion}, we conclude by discussing how our result complements the state of the art.

\section{The exponential polynomial solution and generating functions}
\subsection{The exponential polynomial solution}
\label{section:exppoly}
The polynomial $\chi_c(z) = z^k + c_{k-1}z^{k-1} + \dots + c_1z + c_0$ is called the characteristic polynomial of the recurrence or differential equation. An initialisation $v$, determines a unique solution to the recurrence, which is of the form
\begin{equation}
\label{eq:exppolyds}
u[n] = \sum_i\sum_{j=0}^{m_i - 1}a_{ij}n^j \lambda_i^n
\end{equation}
and a unique solution to the differential equation, which is of the form
\begin{equation}
\label{eq:exppolyct}
u(t) = \sum_i\sum_{j=0}^{m_i - 1} b_{ij}t^j e^{\lambda_jt}
\end{equation}
where $\lambda_i$ is a root of $\chi_c$ with multiplicity $m_i$. For technical reasons, we an alternate basis of polynomials to express $\sum_j a_{ij}n^j$. We rewrite equation \ref{eq:exppolyds} as
\begin{equation}
\label{eq:exppolydscombi}
u[n] = \sum_i\sum_{j=0}^{m_i - 1}b_{ij}\binom{n+j}{j} \lambda_i^n
\end{equation}
where $\binom{n}{r}$ denotes the coefficient of $x^r$ in the binomial expansion of $(1+x)^n$. We observe that the leading coefficients in equations \ref{eq:exppolyds} and \ref{eq:exppolydscombi} have the same sign: they are related as 
$$
b_{i(m_i-1)} = (m_i-1)! \cdot a_{i(m_i-1)}.
$$
We follow \cite[Section 2]{neumann2023} in giving a high level explanation of when an instance of Ultimate Positivity is a robust YES instance, and when it is a robust NO instance. To argue that these sets of instances are semialgebraic, the main ingredient is the quantifier elimination the First Order Theory of the Reals admits. We must therefore encode each of the conditions involved as First Order Formulae. We begin doing so in this subsection.

\subsubsection{Robust YES}
For Ultimate Positivity, it suffices to have a simple real dominant root $\rho$ whose corresponding term in the solution is strictly positive, and asymptotically dominates all the other terms. For a robust YES instance of Ultimate Positivity, this is in fact \textit{necessary}: in case there is another complex root $\lambda$, the magnitude of whose contribution grows as fast as that of $\rho$, the input can be slightly perturbed to obtain a new characteristic polynomial whose dominant root is $\lambda'$, not the real $\rho$. Similarly, if $\rho$ is a repeated root, then one can disturb the input to have a solution whose leading terms oscillate in sign, e.g.\ by replacing the roots $\rho, \rho$ by $\lambda', \bar\lambda'$. See Figure \ref{fig:robustyes} for visual intuition.

\begin{figure}[h]
\centering
\includegraphics[width=0.85\textwidth]{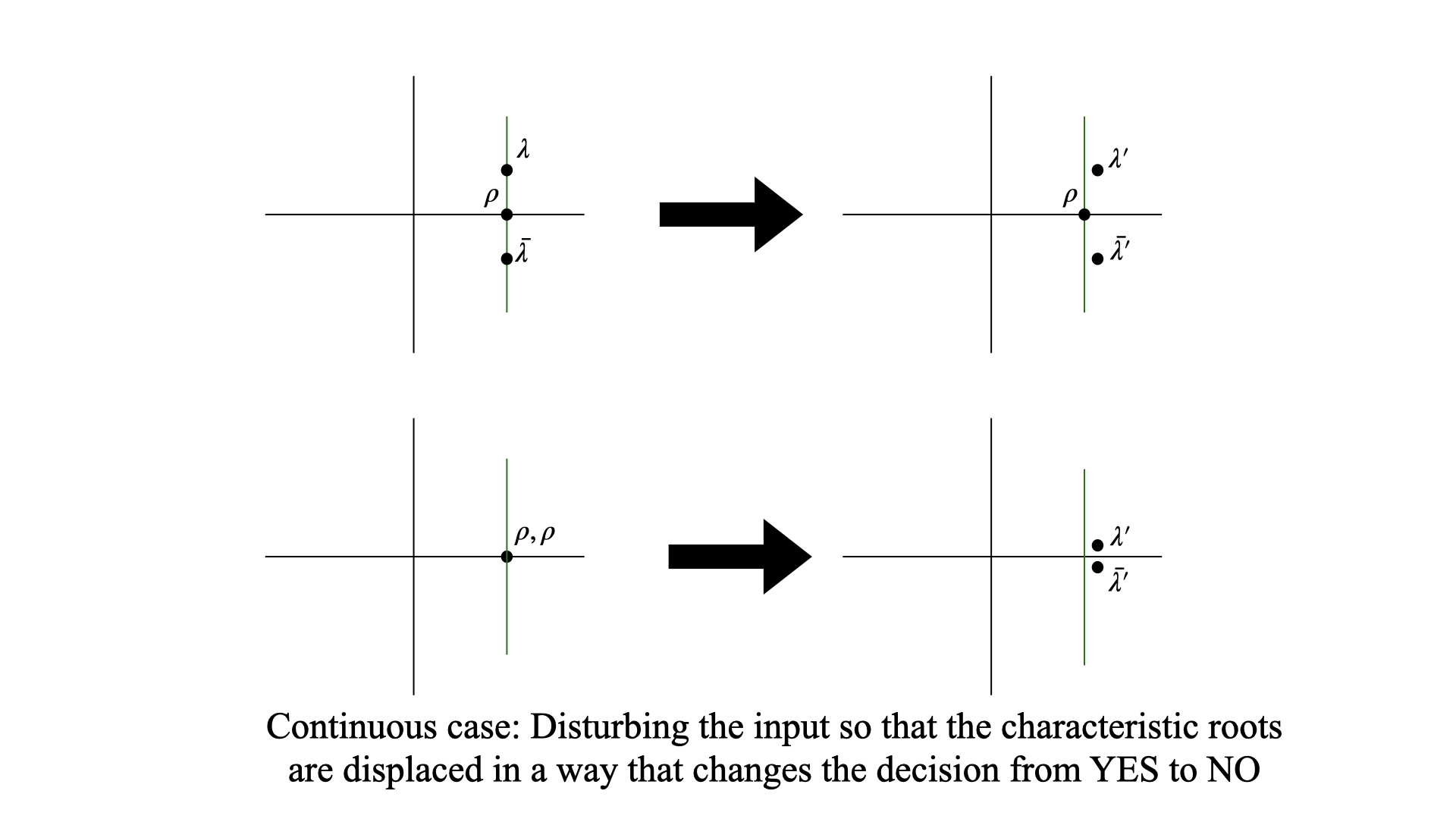}

\caption{The discrete case is conceptually similar.}
\label{fig:robustyes}
\end{figure}

In the discrete case, $\chi_c$ must have a root $\rho$ such that:
\begin{enumerate}
\item $\rho \in \reals$ and $\rho > 0$.
\item $\rho$ has multiplicity $1$.
\item $\rho > |\lambda|$ for all other roots $\lambda$.
\item The coefficient of $\rho^n$ in expression \ref{eq:exppolydscombi} should be strictly positive.
\end{enumerate}
 
Notice the first three conditions are spectral properties, while the last depends on the initialisation. This will be a recurring theme observed in this subsection. Thus, the First Order Formula we seek is of the form
\begin{equation}
\label{eq:robustyes}
\exists \rho.~ \text{spectral}(\rho) ~\land ~\text{initial}(\rho).
\end{equation}
The initial clause (Point 4) is non-trivial, and hence deferred to later sections. The spectral part (Points 1-3), on the other hand, is easy to encode:
$$
\chi_c(\rho) = 0~\land~ \rho > 0  ~\land~ \chi_c^{(1)}(\rho) \ne 0 ~\land~ (\forall x, y.~ \chi_c(x + yi) = 0 \Rightarrow x^2 + y^2 < \rho^2).
$$
Note that we will actually reason about complex input to $\chi_c$ by expanding each of the binomial terms $(x + yi )^k$, and observing that the real and imaginary parts of $\chi_c(x+yi)$ are real valued polynomials\footnotemark ~in $x, y$. The requirements in the continuous case are similar:
\begin{enumerate}
\item $\rho \in \reals$.
\item $\rho$ has multiplicity $1$.
\item $\rho > \Re(\lambda)$ for all other roots $\lambda$.
\item The coefficient of $e^{\rho t}$ in expression \ref{eq:exppolyct} should be strictly positive.
\end{enumerate}

\footnotetext{For example, $c_k z^k + c_{k-1}z^{k-1} + \dots + c_0 = 0$ can be written as $c_k(x + yi)^k + c_{k-1}(x+ yi)^{k-1} + \dots + c_0 = 0$. On expanding each of these binomial terms, we get $$\left(\sum_{j=0}^k c_j \sum_{\ell=0}^{j/2} (-1)^\ell\binom{j}{2\ell}x^{j-2\ell} y^{2\ell}\right) + i\left(\sum_{j=0}^k c_j \sum_{\ell=0}^{j/2} (-1)^\ell\binom{j}{2\ell+1} x^{j-2\ell-1} y^{2\ell+1}\right) = 0.$$ Our desired First Order Theory of the Reals formula is obtained by equating each of the real parenthesised summations to zero.}

\subsubsection{Robust NO}
Ultimate Positivity fails when the leading term oscillates, or is strictly negative. For robust NO instances, we must describe necessary and sufficient conditions for when the above holds despite perturbations to the input. The overall intuition is summarised in Figure \ref{fig:robustinstances}.

\begin{figure}[h]
\centering
\includegraphics[width=0.85\textwidth]{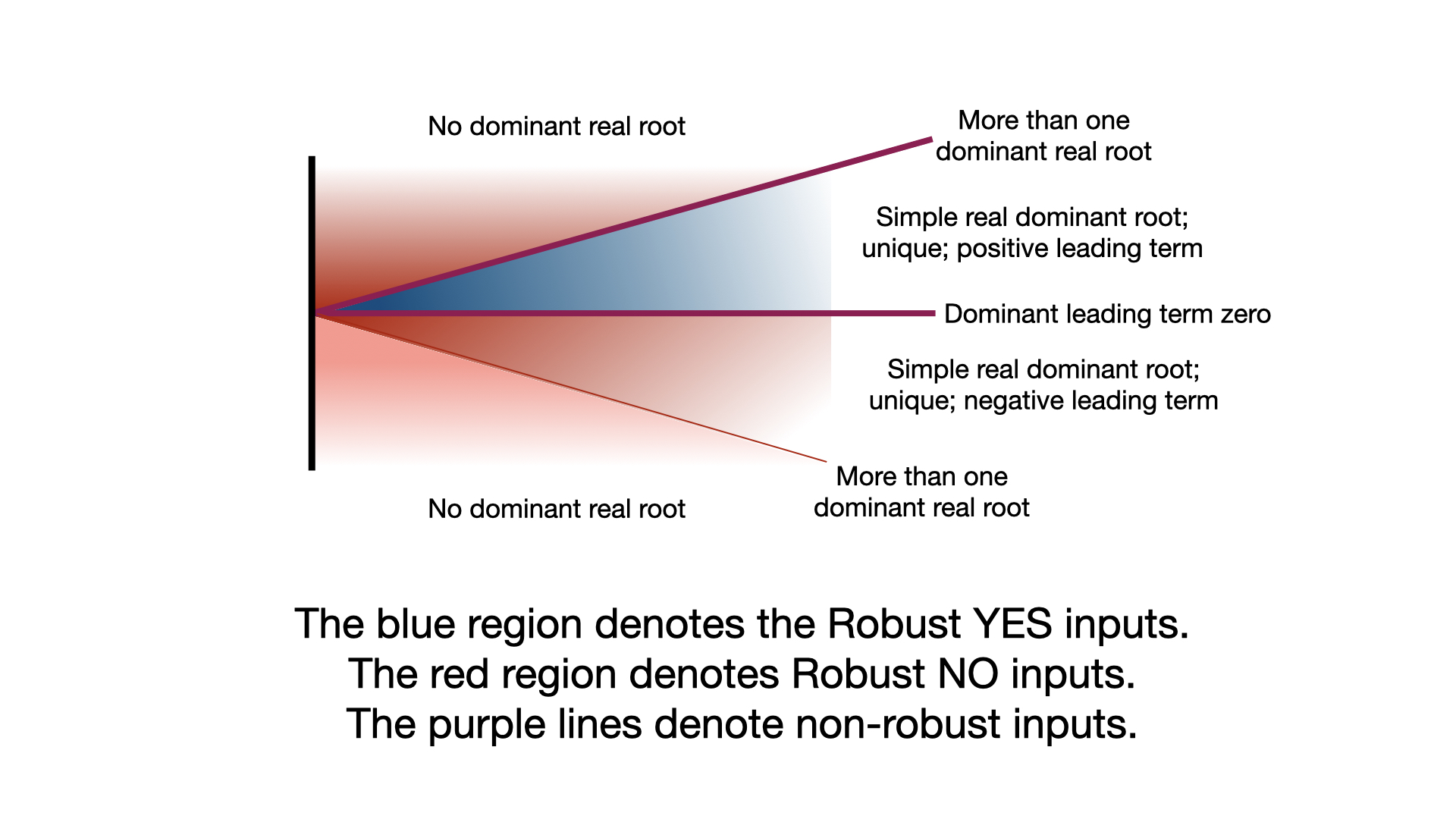}

\caption{Overall intuition.}
\label{fig:robustinstances}
\end{figure}

In the discrete case, at least one of the following two conditions must be met:
\begin{enumerate}
\item The leading term oscillates when $\chi_c$ has a root $\lambda$ such that
\begin{enumerate}
\item $\lambda$ is real and strictly negative, or complex with a nonzero imaginary part.
\item $|\lambda| > \rho$ for all real positive roots $\rho$.
\item $\lambda^n$ is multiplied with a nonzero polynomial in expression \ref{eq:exppolydscombi}.
\end{enumerate}
\item The leading term is strictly negative, and does not become positive (i.e. remains strictly negative, or oscillates) on disturbing the input when $\chi_c$ has a root $\lambda_1= \rho$ such that
\begin{enumerate}
\item $\rho \in \reals$ and $\rho > 0$.
\item $\rho > \lambda$ for all other real positive roots $\lambda$.
\item The leading coefficient of the polynomial multiplied with $\rho^n$ in expression \ref{eq:exppolydscombi} is strictly negative, i.e. $b_{1(m_1-1)} < 0$.
\end{enumerate}
\end{enumerate}

As we did in formula \ref{eq:robustyes}, we will encode the requirements as
\begin{equation}
\label{eq:robustno}
(\exists x, y. ~\text{spectral}_{1}(x, y) ~\land ~\text{initial}_1(x, y)) ~\lor ~ (\exists \rho.~\text{spectral}_2(\rho)~\land ~\text{initial}_2(\rho))
\end{equation}
and discuss the initial clauses in a later section. It is easy to see
$$
\text{spectral}_1(x, y) \equiv ( x < 0 \lor y \ne 0) ~\land~ (\chi_c(x + yi) = 0) ~\land ~(\forall \rho > 0. ~\chi_c(\rho) = 0 \Rightarrow x^2 + y^2 > \rho^2)
$$
and 
$$
\text{spectral}_2(\rho) \equiv \rho > 0 ~\land~ \forall \lambda.~ (\chi_c(\lambda) = 0 \Rightarrow \lambda \le \rho).
$$

Similarly in the continuous case, at least one of the following two conditions must be met:
\begin{enumerate}
\item The leading term oscillates when $\chi_c$ has a root $\lambda$ such that
\begin{enumerate}
\item $\lambda$ is complex with a nonzero imaginary part.
\item $\Re(\lambda) > \rho$ for all real positive roots $\rho$.
\item $e^{\lambda t}$ is multiplied with a nonzero polynomial in expression \ref{eq:exppolydscombi}.
\end{enumerate}
\item The leading term is strictly negative, and does not become positive (i.e. remains strictly negative, or oscillates) on disturbing the input when $\chi_c$ has a root $\lambda_1= \rho$ such that
\begin{enumerate}
\item $\rho \in \reals$.
\item $\rho > \lambda$ for all other real roots $\lambda$.
\item The leading coefficient of the polynomial multiplied with $e^{\rho t} $ in expression \ref{eq:exppolydscombi} is strictly negative, i.e. $b_{1(m_1-1)} < 0$.
\end{enumerate}
\end{enumerate}

\subsection{Generating Functions}
\label{section:genfunc}
We now use well-known control-theoretic generating functions for a transformed perspective to encode the remaining initial conditions in the First Order Theory of the Reals.
\begin{definition}[Z-transform]
For a discrete sequence $u[n]$, its Z-transform $U(z) = \mathcal{Z}[u[n]]$ is defined as
$$
U(z) = \sum_{n=0}^\infty z^{-n}u[n]
$$
\end{definition}

\begin{definition}[Laplace transform]
For a continuous function $u(t)$, its Laplace transform $U(z) = \mathcal{L}[u(t)]$ is defined as
$$
U(z) = \int_0^\infty e^{-zt}u(t)dt
$$
\end{definition}

The Z-transform and Laplace transform operate \textit{linearly} on functions. In the following lemmata, we recall some elementary properties of these generating functions that make them useful in our setting.

\begin{lemma}[Properties of the Z-transform]
Let $u[n]$ be a discrete sequence, and let $U(z) = \mathcal{Z}[u[n]]$
\begin{itemize}
\item $$\mathcal{Z}[u[n+k]] = z^kU(z) - z^k\sum_{j=0}^{k-1}z^{-j} u[j]$$
\item $$\mathcal{Z}\left[\binom{n+m}{m}a^n\right] = \frac{z^{m+1}}{(z - a)^{m+1}}$$
\end{itemize}
\end{lemma}

\begin{lemma}[Properties of the Laplace transform]
Let $u(t)$ be analytic over $\reals$, and let $U(z) = \mathcal{L}[u(t)]$
\begin{itemize}
\item $$\mathcal{L}[u^{(k)}(t)] = z^kU(z) - \sum_{j=1}^{k}z^{k-j} u^{(j-1)}(0)$$
\item $$\mathcal{L}\left[t^me^{at}\right] = \frac{m!}{(z - a)^{m+1}}$$
\end{itemize}
\end{lemma}

On applying the first bullets of the respective lemmata on $(c, v)$: the recurrence or differential equation and the initialisation, we observe that
\begin{equation}
\label{eq:Ztf}
\mathcal{Z}[u[n]] = \frac{\psi_{c,v}(z)}{\chi_c(z)}
\end{equation}
\begin{equation}
\label{eq:Ltf}
\mathcal{L}[u(t)] = \frac{\phi_{c,v}(z)}{\chi_c(z)}.
\end{equation}
The computation of the coefficients $\psi_{c,v}$ and $\phi_{c,v}$ from $(c, v)$ is thus straightforward. On applying the second bullets of the respective lemmata on the exponential polynomial expressions \ref{eq:exppolydscombi} and \ref{eq:exppolyct}, we observe that
\begin{equation}
\label{eq:Ztfalt}
\mathcal{Z}[u[n]] = \sum_{i}\sum_{j=0}^{m_i - 1} \frac{b_{ij}z^{j+1}}{(z-\lambda_i)^{j+1}}
\end{equation}

\begin{equation}
\label{eq:Ltfalt}
\mathcal{L}[u(t)] = \sum_{i}\sum_{j=0}^{m_i - 1} \frac{b_{ij}\cdot j!}{(z-\lambda_i)^{j+1}}
\end{equation}

\section{Encoding leading term requirements in First Order Logic}
\label{section:leadterm}
\subsection{Coefficient sign}
We first express the clauses $\text{initial}(\rho)$ from formula \ref{eq:robustyes} and $\text{initial}_2(\rho)$ from formula \ref{eq:robustno}. Indeed, a direct corollary of the following lemma is that 
\begin{equation}
\text{initial}(\rho) \equiv \psi_{c,v}(\rho) > 0 ;~
\text{initial}_2(\rho) \equiv \psi_{c,v}(\rho) < 0
\end{equation}
in the discrete case, and
\begin{equation}
\text{initial}(\rho) \equiv \phi_{c,v}(\rho) > 0; ~
\text{initial}_2(\rho) \equiv \phi_{c,v}(\rho) < 0
\end{equation}
in the continuous case.
\begin{lemma}
Let $\lambda_1 = \rho$ be a real strictly positive (resp.\ real) root of $\chi_c$ such that for all other real roots $\lambda$ of $\chi_c$, $\rho > \lambda$. Then, the leading term coefficient $b_{1(m_1-1)}$ in expressions \ref{eq:exppolydscombi} (resp.\ \ref{eq:exppolyct}) has the same sign as $\psi_{c,v}(\rho)$ (resp.\ $\phi_{c,v}(\rho)$).
\end{lemma}
\begin{proof}
We denote $b_{10}, \dots, b_{1j}, \dots, b_{1(m_1-1)}$ as $b_0, \dots, b_j, \dots, b_{m-1}$. We can simplify expressions \ref{eq:Ztfalt} and \ref{eq:Ltfalt} as 
\begin{equation}
\label{eq:decD}
\mathcal{Z}[u[n]] = \frac{b_{m-1}z^m}{(z-\rho)^m} + \dots + \frac{b_0z}{z-\rho} + \frac{p(z)}{q(z)}
\end{equation}

\begin{equation}
\label{eq:decC}
\mathcal{L}[u[n]] = \frac{b_{m-1}(m-1)!}{(z-\rho)^m} + \dots + \frac{b_0}{z-\rho} + \frac{p(z)}{q(z)}
\end{equation}

Here $(z-\rho)^mq(z) = \chi_c(z)$. We note that in either case, $q(z)$ is a monic polynomial, each of whose real roots is less than $\rho$. Thus, $q(\rho) > 0$. 

We can add the decomposed fractions together by taking the least common denominator, and compare with expressions \ref{eq:Ztf} and \ref{eq:Ltf}
\begin{equation}
\frac{\psi_{c,v}(z)}{\chi_c(z)} = \frac{q(z)(b_{m-1}z^m + b_{m-2}z^{m-1}(z-\rho)+\dots+b_0z(z-\rho)^{m-1}) + p(z)(z-\rho)^m}{\chi_c(z)}
\end{equation}
\begin{equation}
\frac{\phi_{c,v}(z)}{\chi_c(z)} = \frac{q(z)(b_{m-1}(m-1)! +\dots+b_0(z-\rho)^{m-1}) + p(z)(z-\rho)^m}{\chi_c(z)}
\end{equation}
It is thus easy to see that
\begin{equation}
\psi_{c,v}(\rho) = b_{m-1}\rho^m q(\rho)
\end{equation}
Here, we have established the positivity of $q(\rho)$, and $\rho$ is positive by premise. Thus the sign of $\psi_{c,v}(\rho)$ is the same as that of $b_{m-1}$. Similarly
\begin{equation}
\phi_{c,v}(\rho) = b_{m-1} (m-1)! q(\rho)
\end{equation}
By the same argument, the sign of $\phi_{c,v}(\rho)$ is the same as that of $b_{m-1}$.
\end{proof}

\subsection{Non-zero polynomial testing}

We complete the First Order encoding by showing how to express $\text{initial}_1(x, y)$ from formula \ref{eq:robustno}. Observe the structure of expressions \ref{eq:decD} and \ref{eq:decC}. Let root $\lambda$ have multiplicity $m$ and let its coefficients in the exponential polynomial solution be $b_0, b_1, \dots, b_{m-1}$. These coefficients are all $0$ if and only if the gcd of $\chi_c$ with $\psi_{c,v}$ (resp.\ $\phi_{c,v}$) is divisible by $(z - \lambda)^m$. For example
\begin{equation}
\label{eq:decD2}
\mathcal{Z}[u[n]] = \frac{b_{m-1}z^m}{(z-\lambda)^m} + \dots + \frac{b_0z}{z-\lambda} + \frac{p(z)}{q(z)} = \frac{p'(z)}{q'(z)} = \frac{\psi_{c,v}(z)}{\chi_c(z)}.
\end{equation}
We need to detect the \textit{negation}, that is, at least one of the coefficients is nonzero. In this case, $q'(z)$ will have $(z-\lambda)$ as a factor. This means that on cancelling out all the common factors with $\psi_{c, v}$ (or $\phi_{c, v}$ in the continuous case) from $\chi_c$, a factor of $(z-\lambda)$ still survives.

We use $*$ to denote the convolution of coefficient vectors, i.e. polynomial multiplication. We implicitly use the bounded degree of the recurrence to hardcode and define convolution and mimic evaluation at complex points. Our First Order formula $\text{initial}_1(x, y)$ has the structure:
\begin{align*}
\exists g, h_1, h_2~&. \\
\psi_{c, v} = g * h_1 ~&\land \\
\chi_c = g * h_2 ~\land 
h_2(x + yi) = 0 ~&\land \\
(\forall f. (\exists f_1, f_2. ~ \psi_{c,v} = f * f_1 \land \chi_c = f * f_2) \Rightarrow (\exists \ell. ~g = f*\ell)
\end{align*}
The above asserts that dividing $\chi_c$ by its gcd $g$ with $\psi_{c, v}$ still retains $\lambda = x + yi$ as a root of the quotient polynomial $h_2$, giving us our desired nonzero coefficients.

\section{Conclusion and Perspective}
\label{section:discussion}
The recent attention to robust variants of decision problems related to Linear Dynamical Systems had led to significant strides in circumventing the number-theoretic hardness that the conventional approach is fraught with, as well as in making decision procedures feasible. Indeed, the robust Ultimate Positivity problems \cite{originalarxiv} considers are shown decidable in $\mathsf{PSPACE}$. In more related work, \cite{neumann2023} works in the bit-model of real computation with the notion of robustness considered in this paper. It establishes a notion of polynomial-time efficiency for the detection of Robust YES and NO instances of Ultimate Positivity. It is our hope that our work serves to strengthen this result by offering it theoretical backing and more accessible intuition. We argue that this is indeed useful: for instance, in view of the semialgebraic characterisation, one can easily generalise the detection procedure to work on other models of real computation, e.g.\ the Blum-Shub-Smale model \cite{blumshubsmale}.
\clearpage
\bibliography{main}
\clearpage

\appendix

\end{document}